\documentclass{article}
\usepackage{graphicx}
\usepackage{physics}
\graphicspath{ {/} }
\usepackage{PRIMEarxiv}
\usepackage[ruled,vlined]{algorithm2e}
\usepackage[utf8]{inputenc} 
\usepackage[T1]{fontenc}    
\usepackage{hyperref}       
\usepackage{url}            
\usepackage{booktabs}       
\usepackage{amsfonts}       
\usepackage{nicefrac}       
\usepackage{microtype}      
\usepackage{lipsum}
\usepackage{fancyhdr}       
\usepackage{graphicx}       
\graphicspath{{media/}}     
\usepackage{amsthm}
\usepackage{caption}
\usepackage{subcaption}
\usepackage[english]{babel}
\newtheorem{theorem}{Theorem}[section]

\newtheorem{dec}{Decision Problem}[section]
\newtheorem{proposition}{Proposition}[section]
\newtheorem{statement}{Statement}[section]
\newtheorem{question}{Question}[section]
\theoremstyle{definition}
\newtheorem{definition}{Definition}[section]

\title{TURING MACHINES CANNOT SIMULATE THE HUMAN MIND
\thanks{\textit{\underline{Citation}}: 
\textbf{Abhinav Muraleedharan. Turing Machines cannot simulate the human mind.}} 
}

\author{
  Abhinav Muraleedharan \\
  University of Toronto \\
  \texttt{\{Abhinav.Muraleedharan@mail.utoronto.ca\}} \\
}

\begin{document}
\maketitle

\begin{abstract}
Can a Turing Machine simulate the human mind? If the Church-Turing thesis is assumed to be true, then a Turing Machine should be able to simulate the human mind. In this paper, I challenge that assumption by providing \textbf{strong} mathematical arguments against the Church-Turing thesis. First, I show that there are decision problems that are computable for humans, but uncomputable for Turing Machines. Next, using a thought experiment I show that a humanoid robot equipped with a Turing Machine as the control unit cannot perform all humanly doable physical tasks. Finally, I show that a quantum mechanical computing device involving sequential quantum wave function collapse can compute sequences that are uncomputable for Turing Machines. These results invalidate the Church-Turing thesis and lead to the conclusion that the human mind cannot be simulated by a Turing Machine. Connecting these results, I argue that quantum effects in the human brain are fundamental to the computing abilities of the human mind.
\end{abstract}
\section{Introduction}
Turing Machines \cite{turing1936a} can certainly simulate some aspects of the human mind. Computers can defeat humans in chess \cite{campbell2002deep}, Go \cite{silver2017mastering} and state of the art AI systems can can recognize images \cite{krizhevsky2012imagenet}, create art \cite{ramesh2022hierarchical}, write poems \cite{brown2020language} and manipulate objects \cite{gao2021kpam}. But, is it possible to build a Turing Machine that is exactly identical to the human mind? Can all properties of the human mind including conciousness, general intellingence be replicated in a Turing Machine? Is there a fundamental difference between the human mind and Turing Machines?\\

Turing himself argued that digital computers can be made indistinguishable from humans \cite{10.1093/mind/LIX.236.433}. In contrast to this, Penrose, in his seminal works \cite{penrose1989emperor, penrose1994shadows}
has proposed that quantum coherence and wave function collapse process is connected to conciousness and computation in human mind is fundamentally different from that of a Turing Machine. Similar to arguments from Lucas \cite{lucas1961minds}, from certain deductions from G\"{o}del's incompletness theorem, Penrose argued that human \textit{understanding} or \textit{conscious thought} is non-algorithmic. Together with Hameroff, Penrose further proposed the Orchestrated Objective Reduction Theory \cite{hameroff1996conscious, hameroff1996orchestrated, hameroff2014consciousness} wherein quantum coherence in microtubules and  objective reduction of the quantum wave function plays a central role in the functioning of the human mind. Similar to Penrose's arguments, Stapp also proposed a quantum model of the mind body system \cite{stapp2004mind, stapp2000importance} with parallel computation aspect governed by the Schr\"{o}dinger equation and action selection governed by quantum wave function collapse. Stapp argued that quantum zeno effect \cite{misra1977zeno} in the brain would prevent environmental decoherence and proposed connection between quantum zeno effect and conscious control.\\

 These arguments were later challenged by Tegmark. In \cite{tegmark2000importance}, Tegmark has shown that the decoherence time scales are typically shorter than relevant dynamical time scales in brain. Hence Tegmark argued that quantum coherence is not related to consciousness and brain is a classical computing system. Although results from Tegmark indicate that quantum coherence would not be preserved in brain for longer timescales, these results cannot be used to rule out quantum effects in the brain completely. To make strong statements about quantum effects and its role in computing capabilities of the brain, one would have to answer the following question:
 
\begin{itemize}
    \item Is a Turing Machine mathematically equivalent to the human mind ? 
\end{itemize}
In this paper I study this question from three different aspects:
\begin{itemize}
\item First, I ask if a Turing Machine is equivalent to a human mathematician. That is, can all statements computed by a human mathematician be computed by a Turing Machine? 

\item Secondly, I introduce a novel thought experiment to ask the question: "Is it possible to build a humanoid robot that can perform all possible humanly doable physical tasks in any humanly operable environments?".

\item Finally, I ask if a quantum mechanical system involving sequential quantum wave function collapse can compute sequences that are uncomputable for a Turing Machine.

\end{itemize}
Studying these fundamental questions leads to the conclusion that the collapse of quantum wave function is fundamental to the computing capabilities of human mind. Either the theory put forth by Penrose \cite{penrose1989emperor} or that of Stapp \cite{stapp2000importance} or a different theory incorporating quantum wave function collapse is \textbf{required} to explain the computing capabilities of the human mind. Human brain is not a classical computing system.
\section{Definitions}
\subsection{Turing Machine:}

A Turing Machine \cite{turing1936a} can be formally specified as a quadruple $T = (Q,\Sigma,s,\delta)$, where:
\begin{itemize}
    \item $Q$ is a finite set of states $q$
    \item $\Sigma$ is a finite set of symbols
    \item $s$ is the initial state $s \in Q$
    \item $\delta$ is a transition function determining the next move:\\
        $\delta: (Q \times \Sigma) \rightarrow (\Sigma \times {L,R} \times Q)$
\end{itemize}
In addition to the quadruple of mathematical objects that specify a Turing Machine, a Turing Machine, as described by Turing also contains an infinite tape. Hence the definition of Turing Machine would contain an additional statement:
\begin{enumerate}
    \item The memory tape of a Turing Machine is infinite.
\end{enumerate}

\subsection{Computing Machine:}

In \cite{turing1936a}, Turing uses the term 'computing machine' to describe Turing Machines. In this paper, by computing machine, I imply machines that are similar in construction to a Turing Machine, but with a finite memory tape. A computing machine can be formally specified as a quadruple $T = (Q,\Sigma,s,\delta)$, where:
\begin{itemize}
    \item $Q$ is a finite set of states $q$
    \item $\Sigma$ is a finite set of symbols
    \item $s$ is the initial state $s \in Q$
    \item $\delta$ is a transition function determining the next move:\\
        $\delta: (Q \times \Sigma) \rightarrow (\Sigma \times {L,R} \times Q)$
\end{itemize}
In addition to the quadruple of mathematical objects that specify a Computing Machine, a Computing Machine contains a finite tape, which means that any computing machine is not computationally as powerful as a Turing Machine. Hence the definition of Computing Machine would contain an additional statement:
\begin{enumerate}
    \item The memory tape of a Computing Machine is finite.
\end{enumerate}
For simplicity, the quadruple $T = (Q,\Sigma,s,\delta)$ would be called as a 'program' in this paper.
\section{Can a Turing Machine construct any mathematical object ?}
Inspired by the great works of Russel \cite{russell1902letter}, Godel \cite{godel1931formal}, Turing \cite{turing1936a}, to study the fundamental nature of human mind, we employ two big weapons of modern mathematics: \textit{self reference} and \textit{proof by contradiction}. We construct a decision problem whereby a Turing Machine has to compute a mathematical truth about itself.\\

Given a formal system $F$, a Turing Machine, in principle can compute Non G\"{o}del Type true statements within the formal system $F$. Turing Machines can also check the validity of a proof within a formal system $F$ algorithmically. But can a Turing Machine construct \textit{any} mathematical object? This question is difficult to answer directly, hence instead of asking whether a Turing Machine can construct any mathematical object, let's ask: " For any mathematical object constructed by a human mathematician, can a Turing Machine compute all statements associated with the definition of of the mathematical object?". If a Turing Machine can do so, then it implies that Turing Machines are equivalent to human mathematicians, and therefore a Turing Machine should be able to construct new mathematical objects just like human mathematicians.\\

First, let's assume that Turing Machines are equivalent to human mathematicians, hence all mathematical objects constructed by a human mathematician could in principle be constructed by a Turing Machine. This implies that all statements associated with the definition of any mathematical object should be computable by a Turing Machine.\\

Any mathematical object $MO$ can be fully defined using a finite set of statements $P = \{P_1, P_2,P_3,....P_k\}$.\\
Corresponding to each statement $P_k$ in the definition of an arbritrary mathematical object $MO$, one can formulate a decision problem $D_k$ as follows:
\begin{itemize}
    \item $D_k:$Is $P_k$ true for $MO$ ?
\end{itemize}
Let set $D =  \{D_1, D_2,D_3,....D_k\}$ be the collection of all such decision problems $D_k$. As per our assumption, all decision problems in set $D$ should be computable for any mathematical object $MO$. The answer to these decision problems should be 'yes', and there should exist an algorithmic method to compute the correct answer.\\

Now let's construct two mathematical objects as per definition in section 2:
\begin{enumerate}
    \item A Computing Machine $CM$ of finite memory $M$.
    \item A Turing Machine $TM$ of infinite memory.
\end{enumerate}
For computing machine $CM$, one of the statements that define $CM$ is:
\begin{statement}
     Computing machine $CM$ has finite memory $M$.
\end{statement}
Corresponding to \textbf{Statement 3.1}, let's formulate a decision problem:
\begin{dec}
   Is your memory = $M$?
\end{dec}
Note that the decision problem is self referential, as the mathematical object in this case is the computing machine itself. By computing answer to Decision Problem 3.1, the computing machine $CM$ computes a mathematical truth about itself. To compute correct answer to Decision Problem 3.1, an algorithm should be able to correctly compute memory of a computing machine. To accurately compute memory of an arbitrary computing machine, the program would have to examine the entire tape, and count the number of memory units in its tape. Consider the following algorithm, which starts with a blank input tape.\\
\begin{algorithm}[H]
\SetAlgoLined
\textbf{Require: }
 $Tape$ $T_n$: Blank Memory Tape\\
 \uIf{Scanned Cell == blank}{
    write 1 \;
  }

 \caption{Computing memory of finite memory computing machine}
\end{algorithm}
If the scanned cell is blank, Algorithm 1 writes down '\textbf{1}' on the scanned cell. For a finite memory computing machine, Algorithm 1 halts after all cells are scanned by the computing machine. Once the program halts, the tape of computing machine $CM$ can be copied to input of a Turing Machine and the number of '\textbf{1}s' in the tape can be counted.\\
\\
For a program which examines the tape, the time taken for computation of memory grows linearly as per the size of memory tape. The total time required $T$ is proportional to the memory $M$ of computing machine $CM$.
\begin{equation}
    T = kM
\end{equation}
where 'k' is some proportionality constant. For any Computing Machine $CM$, $T$ is always going to be a finite value. Hence, Decision Problem 3.1 is computable.\\

Now consider the second mathematical object, a Turing Machine $TM$ of infinite memory. The statement that defines a Turing Machine $TM$ is:
\begin{statement}
     Turing Machine $TM$ has infinite memory $M$.
\end{statement}
Corresponding to \textit{Statement 3.2}, let's formulate a decision problem:
\begin{dec}
   Is your memory infinite ?
\end{dec}
The decision problem is again self referential, as the mathematical object in this case is the Turing Machine itself. By computing answer to Decision Problem 3.2, the Turing Machine $TM$ computes a mathematical truth about itself.\\
If Algorithm 1 is implemented to compute the memory of Turing Machine, it would never halt as there are infinite memory cells to scan in a Turing Machine.
As per eq(1), the total time taken to compute memory is:
\begin{equation}
    T = \infty
\end{equation}
Hence Algorithm 1 cannot be used to compute answer to decision problem 3.2. Is there any other program that can correctly answer decision problem 3.2 in finite time? Consider Algorithm 2, which is basically a lookup table based program, where the mathematical truth is pre-programmed. Algorithm 2 does not scan the memory tape of Turing Machine, hence would halt in finite time. When implemented on a Turing Machine, and given the input: "Is your memory infinite", the output would be: "Yes".
\begin{algorithm}[H]
\SetAlgoLined
\textbf{Require: }
 $Tape$ $T_n$: Input\\
 \uIf{Input= 'Is your memory infinite?'}{
    Output = 'Yes, my memory is infinite!' \;
  }
 \caption{Lookup Table Algorithm}
\end{algorithm}
However a lookup table based algorithm can return false statements. The same algorithm would output: 'Yes, my memory is infinite' when implemented on any computing machine with sufficient finite memory, leading to a contradiction. For instance, consider the situation where one implements Algorithm 2 in a Computing Machine $CM$ of finite memory $M$. In this case, when given the input: "Is your memory infinite ? ", it would return: "Yes, my memory is infinite". This output is clearly a false statement, as the program is implemented on a computing machine of finite memory. 
 
 \begin{theorem}
No program can compute correct answer to the decision problem: "Is your memory infinite ?"  without examining the memory tape of an arbitrary computing machine. 
 \begin{proof}
(Proof by Contradiction.)
 Let there exist a program $p$, which in finite time, performs some computations and outputs a statement $O$ to the question: "Is your memory infinite ?" without examining the entire tape. Let's assume that output $O$ is always true. \\
 \\
 Such a program, by definition should always halt, hence it utilizes finite amount of memory. Consider the following thought experiment: I implement the program $p$ on a Turing Machine, and observe its functioning. I determine the maximum amount of memory, $M$ required for execution of the program.  Since the program is assumed to output true statements, when implemented on a Turing Machine, it returns: " Yes my memory tape is infinite " as the output.
 Without making any modifications to the program, or to any finite input that is required by the program, I run the exact same program on a Computing Machine $CM$ of finite memory $M$. As I have not made any modifications to the program, and since the program is deterministic, the program outputs the same statement: "Yes, my memory tape is infinite" as the output. This output is a false statement, as the memory of Computing Machine $CM$ is finite. As the program do not examine the tape, it cannot 'detect' this modification. Hence, this leads to a contradiction, and it implies that our assumption: "The output $o$ of $p$ is always true" is false. Therefore, no program, in finite time can compute correct answer to the question "Is your memory infinite ?" without examining the tape. To compute the correct answer, the program should examine the entire tape, and a program that examines the entire tape would not halt on a Turing Machine.

 \end{proof}
 \end{theorem}
 
 Hence, the decision problem: "Is your memory infinite?" is uncomputable for a Turing Machine. However, this decision problem is computable for the human mathematician who constructed the Turing Machine. When constructing a mathematical object called the Turing Machine, the mathematician computes a statement that is uncomputable to the Turing Machine itself. A Turing Machine cannot 'know' its memory is infinite. Therefore a Turing Machine cannot compute a mathematical truth about itself. This contradicts our assumption that Turing Machines are equivalent to human mathematicians. Therefore, our initial assumption is wrong; hence there is a fundamental difference between computation in the human mind and the Turing Machine. 
 \section{Connection to the Halting Problem }

There is a connection between computation of the statement: "A Turing Machine has infinite memory" and the halting problem. In \cite{turing1936a}, Turing showed that the halting problem is uncomputable for programs with blank input tape.

Hence, there is no algorithm to classify programs in the set of all programs $P$ into two distinct sets $P_{halt}$ and $P_{\infty}$  where $P_{halt}$ is the set of all programs that halt on any Turing Machine and $P_{\infty}$ is the set of all programs that don't halt on a Turing Machine.

\subsection{Classification of Programs}
The set of all Programs $P$ can be divided into 3 mutually exclusive sets,$P_1, P_2, P_3$ based on motion of the tape of the Turing Machine.
\begin{equation}
    P_1 \cup P_2 \cup P_3 = P
\end{equation}
 $P_1$, $P_2$, $P_3$ are defined based on the following properties.
\begin{enumerate}
  \item $\forall p \in P_1, p$ \emph{halts} in finite time, when programmed on a Turing Machine $TM$. $\forall p \in P_1, p$ also \emph{halts} when programmed on a computing machine $CM$ with finite memory tape of sufficient memory.\\
  \item $\forall p \in P_2, p$, \emph{does not halt} in any Turing Machine $TM$ and the memory tape of any Turing Machine $TM$ gets into perfect periodic motion. $\forall p \in P_2, p$ \emph{does not halt} in a computing machine $CM$ with finite memory tape of sufficient memory.\\
  
  \item $\forall p \in P_3, p$ \emph{does not halt} in a Turing Machine $TM$, and consumes infinite memory during execution of $p$. In this case, the Turing Machine $TM$ enters into a non-periodic non halting behaviour.  $\forall p \in P_3, p$ \emph{halts} in a computing machine $CM$ with any finite memory tape. This is because the program eventually consumes the entire memory of Computing Machine $CM$, due to its non periodic behaviour.
\end{enumerate}
\subsection{Halting Problem}
A program that does not halt on the Turing Machine may halt on a Computing Machine of finite memory. Hence the set of programs that halt and those do not halt are different for Turing Machines and Computing Machines of finite memory.
Let $H_{TM}$ be the set of all programs that halt on a Turing Machine $TM$.
\begin{equation}
H_{TM} = P_1 
\end{equation}
Let $H_{TM}'$ be the set of all programs that does not halt on a Turing Machine $TM$.

\begin{equation}
H_{TM}' = P_2 \cup P_3 
\end{equation}
For a computing machine $CM$ with finite memory tape $M$, the sets of programs that halt and does not halt are different from that of Turing Machine $TM$.\\
For a computing machine $CM$ with finite memory $M$, let $H_{CM}$ be the set of all programs that halt.
\begin{equation}
H_{CM} = P_1 \cup P_3 
\end{equation}
Let the programs that do not halt on a computing machine $CM$ be $H_{CM}'$.
\begin{equation}
H_{CM}' = P_2 
\end{equation}
Given any program $p \in P$, lets pose the following decision problem:
 \begin{dec}
Does $p$ halt in $X$ ? 
\end{dec}
Where $X$ can be either a Turing Machine $TM$ of infinite memory or a Computing Machine $CM$ of finite memory $M$. It is meaningless to only ask "Does this program halt" as the halting behaviour of program depends on whether it is executed in a Turing Machine of infinite memory or a Computing Machine with finite memory $M$. It is unknown whether $X = TM$ or $X = CM$. To solve \emph{Decision Problem 4.1}, One would have to perform the following computations.
\begin{enumerate}
    \item Compute to which set ($P_1$ or $P_2$ or $P_3$), $p$ belongs to.
    \item Compute memory of $X$.
\end{enumerate}
Step 1, which is classification of an arbitrary program $p$ into sets $P_1, P_2, P_3 $ can be performed by the following procedure:
\begin{itemize}
    \item Simulate $p$ in a Finite Automaton $FA$ of sufficient finite memory $M$
\end{itemize}
Depending on which set the program $p$ belongs to, one of the following things will happen:
\begin{enumerate}
    \item The program $p$ halts in finite time $t$, but the total memory is not consumed.
    \item The finite automaton repeats its own state, and falls into perfect periodic motion.
    \item The finite automaton halts, with memory full.
\end{enumerate}
\begin{itemize}
 
   \item If 1 happens, then the program $p$ is in set $P_1$.
   
   \item If 2 happens, then the program $p$ is in set $P_2$.
   
   \item If 3 happens, then the program $p$ is in set $P_3$.
\end{itemize}
A natural question would be, How to calculate $M$? If the memory of Finite Automaton $FA$ is not sufficient, some programs in set $P$ might halt because the memory requirement of the program is more than the memory of the $FA$. Hence the classification of $p$ into sets $P_1, P_2, P_3$ would be wrong. Instead of computing required memory for each program $p$, we ask a human to simply predict the memory requirement of $p$.\\
\begin{itemize}
    \item $\forall p\in P$, Let a human $H$ predict $M_p$ from the set of all natural numbers $\mathbb{N}$.
\end{itemize}
\begin{proposition}
If $p \in P_1$ or $P_2$, then the probability that the predicted memory value $M_p$ is greater than the required memory $M_{req}$ of program $p$ is 1.\\
\end{proposition}
\begin{proof}
Let $A = \{1,2,3,4,......N\}$ be a finite set with $N$ elements. $A \subseteq \mathbb{N}$.
Assuming uniform distribution , the probability of a human predicting  a number $n$ greater than $k$ from a finite set $A$ with n elements:
\begin{equation}
    \mathbb{P}(n > k) = \frac{N - k}{N}
\end{equation}
\begin{equation}
   \mathbb{P}(n > k) = 1 - \frac{k}{N}
\end{equation}
As $N \to \infty$ or when $A = \mathbb{N}$, $\mathbb{P}(n > k) = 1$.\\
Given any program $p \in$ $P_1$ or $P_2$, with a memory requirement $M_{req} \in \mathbb{N}$, the probability that the predicted $M_{p}$ is greater than $M_{req}$, $\mathbb{P}(M_p > M_{req} ) = 1$ (As per \emph{eq}(11)).
\end{proof}
Hence the probability that any program $p$ would be correctly grouped into set $P_1$ or $P_2$ or $P_3$ is 1. Let $E_1$ be the event that denotes correct classification of all programs $p \in P$ into sets $P_1$ or $P_2$ or $P_3$. \\
\begin{equation}
\mathbb{P}(E_1) = 1 
\end{equation}
Let $E_2$ be the event corresponding to step 2, which is successful computation of memory of $X$.  
\begin{equation}
\mathbb{P}(E_2) = x
\end{equation}
$\mathbb{P}(E_2)$ depends on $X$, as $X$ can be either Turing Machine of infinite memory or a Computing Machine of finite memory.
Let $E$ be the event that denotes correct classification of all programs $p \in P$ into sets $P_{halt}$ and $P_{\infty}$ for $X$. The total probability of occurrence of $E$ is equal to product of probabilities $\mathbb{P}(E_1)$ and $\mathbb{P}(E_2)$.
\begin{equation}
\mathbb{P}(E) =  \mathbb{P}(E_1)\mathbb{P}(E_2)   
\end{equation}
 Due to uncomputability of halting problem, if $X = TM$ or the Turing Machine, the probability of occurrence of E is equal to zero.
\begin{equation}
\mathbb{P}(E) = 0    
\end{equation}
Which implies that
\begin{equation}
\mathbb{P}(E_2) = x = 0    
\end{equation}
This basically means that the probability of computing statement 2 is $0$. Hence the memory of a Turing Machine is uncomputable.

\section{Thought Experiment: Uncountably Many Environments and Infinite Tasks}
The computing capabilities of human mind are not just limited to the construction of arbitrary mathematical objects. Mathematics, in the history of human civilization is a more recent development, and the biological structure of human brain has remained almost the same since the first humans. \\

Humans have the unique ability to develop sophisticated tools using resources from the environment. Developing sophisticated tools requires humans to be able to perform complex task sequences exploiting the physics of the environment. The ability of humans to perform complex task sequences, in diverse environments is a fundamental feature of human intelligence. If all aspects of human mind can be simulated in a Turing Machine, then it should be possible to construct a humanoid robot that is capable of performing complex humanly doable task sequences.\\

Imagine a humanoid robot \textit{R} with physical capabilities similar to a typical human being. The humanoid robot \textit{R} is equipped with a Turing Machine $TM$ as the 'brain', and $TM$ can simulate any possible programs. Now let's ask the following question:\\
\begin{question}

Can a humanoid robot $R$ perform all humanly doable tasks within any humanly operable environment $E_x$ ?
\end{question}
To answer \textit{Question 5.1}, we need proper definitions of environment and tasks. An environment informally is a region of space containing a collection of physical objects arranged in a definite manner. A task, in an environment is an action sequence performed by an agent that may result in change of positions or orientations of the objects in the environment. Definitions of environment and task is dependent on definition of an object, so let's define an object first.\\
\begin{definition}[Object]
An object $o$ is a 3D solid with definite mass and geometry. Any object can be described with a finite number of parameters. For instance, an object $o$ can be completely defined by the finite set $o = \{p_0,p_1,p_2,....p_k,...p_N \}$ where $p_k$ encodes information defining geometry and mass of the object. Further, every parameter $p_k \in \mathbb{N}$.
\end{definition}
\begin{proposition}
The set of all possible objects $O$ is countably infinite.
\end{proposition}
\begin{proof}
Let $O$ be the set of all objects. $O = \{o_{0},o_{1},o_{2},....o_{i},...o_{N} \}$. Let $o_i$ be any object from set $O$ which can be fully defined by the parameter set $o_i = \{p_{i_0},p_{i_1},p_{i_2},....p_{i_k},...p_{N_i} \}$.\\

For any object $o_i$, a unique G\"{o}del Number $G_{o_i}$ can be assigned to $o_i$, by raising each prime number to the power of each parameter and taking a product of all of them.
\begin{equation}
    G_{o_i} = 2^{p_{i_0}}3^{p_{i_1}}5^{p_{i_1}}.....q_N^{p_{i_N}}
\end{equation}
where $q_N$ is the nth prime number.
Hence any object $o_i$ can be represented by a unique G\"{o}del number $ G_{o_i}$.
\begin{equation}
    G_{o_i} \in \mathbb{N}, \forall o_i \in O
\end{equation}
Therefore the set of all possible objects $O$ is isomorphic to $\mathbb{N}$. Hence $O$ is countably infinite.
\end{proof}

For instance, any polyhedron (see \textit{Figure 1}) with defined dimensions and mass is an object. Any polyhedron can be represented by a finite number of parameters. 
\begin{figure}[ht]
\begin{center}
\includegraphics[scale=0.5]{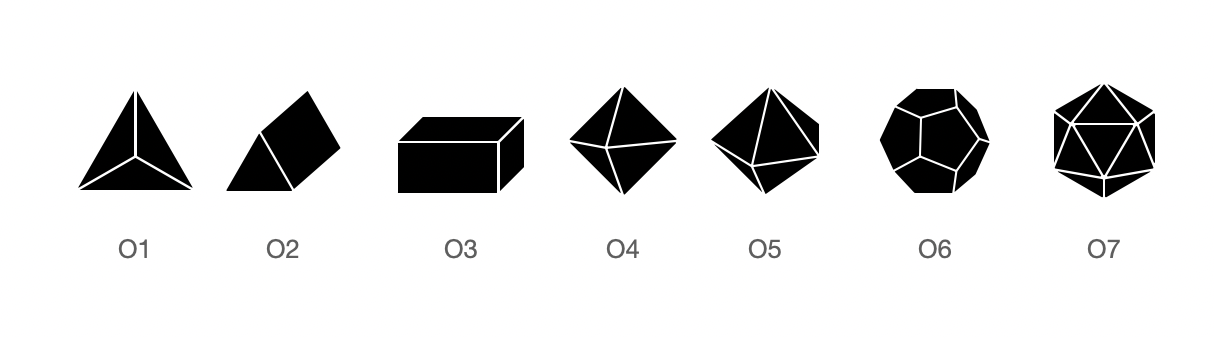}
\end{center}
\caption{A set of polyhedron shaped objects}
\centering
\end{figure}
 Consider the set $O_{poly}$ containing all polyhedron shaped objects. Each object within this set can be indexed as per the number of faces of polyhedron object. Further, the geometry of the object, along with its mass can be represented using a single natural number $G_{o_i}$ $\forall o_i$. The set of all polyhedron objects:
\begin{equation}
    O_{poly} = \{o_1,o_2,o_3,....o_k,...o_N \}
\end{equation}
 $O_{poly}$ is countably infinite and the $N$ th object is a polyhedron with $N+3$ faces.\\
 
Now an environment can be defined. An environment can be constructed by assembling a set of objects in the 3 dimensional space.
\begin{definition}[Environment]
An Environment is a cuboidal region of space with width $w$, depth $d$ and height $h$ containing a collection of objects from the set of all objects $O$. Where $l_{min} \leq w,h \leq l_{max}$, $d \in [d_0, \infty)$ and $l_{max} > l_{min} > 0$. Any environment $E_x$ can be fully defined by the set $E_x = \{O_{E_{x}},C_{E_{x}},P_{E_{x}} \} $ where $O_{E_{x}} = \{o_1,o_2,o_3,....o_k,...o_N \}$ contains the set of objects in environment $E_x$. The set $ C_{E_{x}}$ contains information related to the position and orientation of all objects within environment $E_x$. $ C_{E_{x}} = \{c_1,c_2,c_3,....c_k,...c_N \} $ where $c_k$ contains parameters defining the positional coordinates and relative orientation of any object $o_k$ with respect to a local coordinate system in the environment $E_x$. The set $P_{E_{x}}$ contains parameters that define the cuboidal region of space. $P_{E_{x}} = \{w,h,d\}$.
\end{definition}
\begin{figure}[ht]
\begin{center}
\includegraphics[scale=0.5]{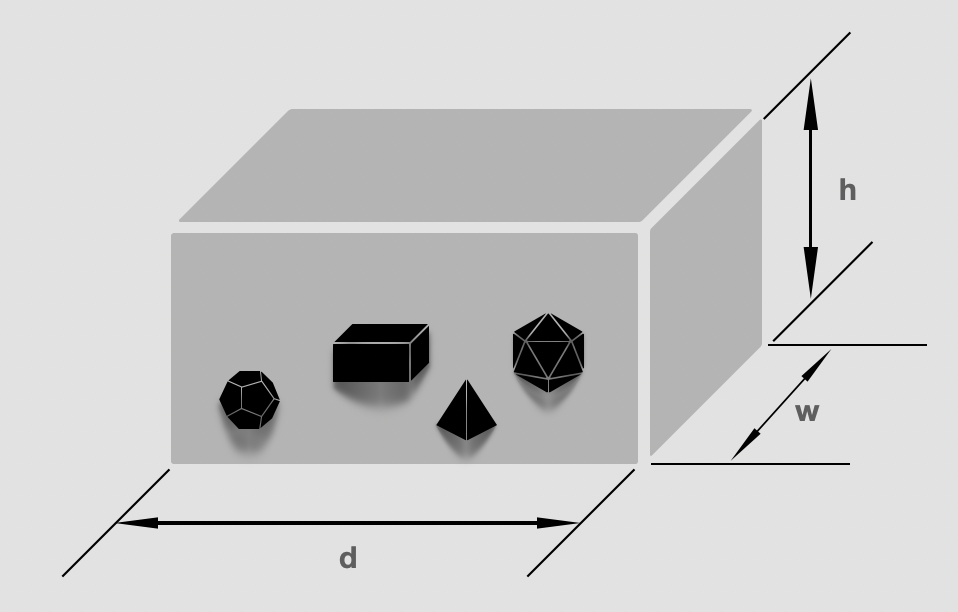}
\end{center}
\caption{An Environment consisting of a collection of objects}
\centering
\end{figure}

\begin{proposition}
The set of all possible environments $E$ is uncountably infinite.
\end{proposition}
\begin{proof}
Let $E$ be the set of all objects. Each environment in $E$ contains a collection of objects from the set of all objects $O$. The set of all environments $E$, hence contains all possible collection of objects from $O$. Therefore, the set of all environments $E$ is isomorphic to the power set of $O$. \\
\begin{equation}
    E \cong \mathcal{P}(O)
\end{equation}
But $O$ is isomorphic to $\mathbb{N}$.
\begin{equation}
  O \cong \mathbb{N}
\end{equation}
Which implies,
\begin{equation}
 E \cong \mathcal{P}(\mathbb{N})
\end{equation}
But power set of natural numbers is isomorphic to the real number set.
\begin{equation}
 \mathcal{P}(\mathbb{N}) \cong \mathbb{R}
\end{equation}
Therefore, the set of all possible environments is isomorphic to the reals.
\begin{equation}
 E \cong \mathbb{R}
\end{equation}
Hence, the set $E$ is uncountably infinite.
\end{proof}
\begin{figure}[ht]
\begin{center}
\includegraphics[scale=0.5]{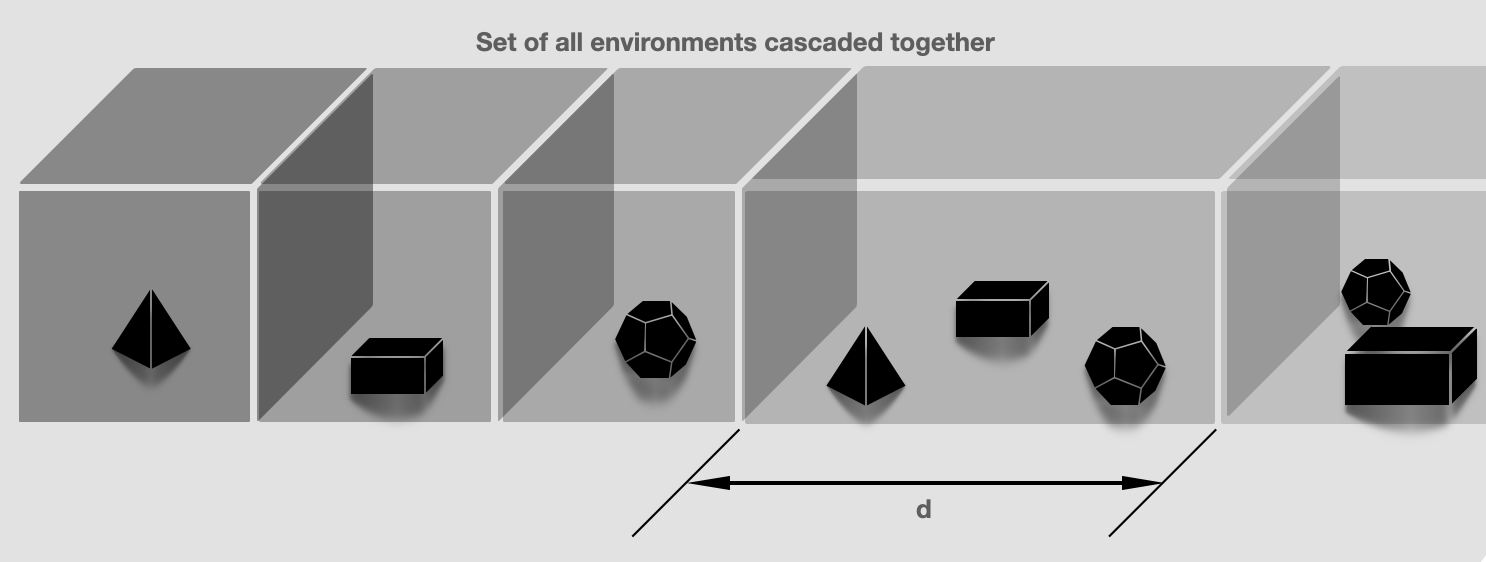}
\end{center}
\caption{Environments with polyhedron objects}
\centering
\end{figure}
In a given environment, an agent (robot or human) can manipulate objects within the environment to change the arrangement of objects. The agent can also move within the environment without affecting the arrangements of objects. To do either of these, the agent would have to perform an action sequence in the environment exploiting the physics of the environment. A finite action sequence performed by an agent that result in the movement of objects within an environment is called as a task. A task(informally) consists of a finite sequence of operations performed on a finite set of objects in any given Environment.
\begin{definition}[Task]
 Let an agent $A$ perform an action sequence $\{a_{0}, a_{1}, a_{2},,,a_{N}\}$ during the time steps $\tau = \{0,1,2,,,N\}$ in an environment $E_x$. The action sequence constitute a task if the configuration of the environment changes at any time steps in $\tau$. 
Let $C_{E_{x}}^{t_k}$ denote the configuration set of the environment $E_x$ at time step $t_k$. If a task is performed in the environment $E_x$ between time steps $t_0$ and $t_N$, then $\exists \{i,j\} \in \tau$, $ i \neq j$ for which $C_{E_{x}}^{t_i} \neq C_{E_{x}}^{t_j}$.
\end{definition}
Any task can be represented with a task instruction code. For instance, high level task instruction codes for a simple pick and place tasks could be: 

\textbf{Task 1:}"Take the laptop from the cupboard and place it on the desk”.\\
\textbf{Task 2:}"Take the laptop from the desk and place it on the cupboard".\\
Tasks can be cascaded together to form a new task. For instance, \textbf{Task 3} can be defined which is a combination of \textbf{Task 1} and \textbf{Task 2}.

$\textbf{Task 3} = \{\textbf{Task 1}, \textbf{Task 2}\}$\\

 Due to actuation limits of the human body, only a subset of the set of all possible tasks in an environment $E_x$ is doable by a human being. However, the set of all possible humanly doable tasks, within an environment with at least one object is countably infinite. This is because objects within an environment can be displaced to arbitrary positions. These displacement tasks can be cascaded together arbitrarily to form new tasks, and hence the set of all possible tasks within an environment is enumerable and countably infinite. Having defined objects, environment and tasks, let's formulate a thought experiment.

\textbf{Thought Experiment:}
From the set of all possible environments, an Environment $E_x$ is selected at random. A humanoid robot $R$ is assigned to perform all humanly doable tasks within the environment $E_x$ in infinite time. We assume that the physical characteristics, actuation limits, sensory apparatus of the humanoid robot is equivalent to a human being, hence it is physically possible for the humanoid robot to perform any humanly doable tasks. Now let's ask the following question:\\

\textbf{Question:} Can the humanoid robot $R$ perform all humanly doable tasks in any environment $E_x$ ?

To perform any task $T_k$, the Turing Machine in the robot's brain should generate action sequences $\{a_i\}_{i=0}^{N_k}$, where $N_k$ is the number of timesteps required for completing task $T_k$. Now consider a set $T_{E_x}$ which consists of an infinite sequence of humanly doable tasks in environment $E_x$. 

\begin{equation}
T_{E_x} = \{T_0,T_1,T_2,....T_{\infty}\}
\end{equation}
To perform all tasks in $T_{E_x}$, the humanoid robot would have to generate an infinite horizon action sequence $A_{E_x} = \{a_i\}_{i=0}^{\infty}$. However, the set of all possible tasks in one environment would be distinct from any other environments.
\begin{equation}
    T_{E_x} \neq T_{E_y}, \forall \{E_x,E_y\}_{x\neq y} \in E
\end{equation}
This implies that the infinite horizon action sequence required to perform all possible tasks in one environment is different from the infinite horizon action sequence required to perform all possible tasks in another environment.
\begin{equation}
    A_{E_x} \neq A_{E_y}, \forall \{E_x,E_y\}_{x\neq y} \in E
\end{equation}

Since the set of all possible environments is uncountably infinite, to perform all humanly doable tasks in any environment, the robot must be capable of generating an uncountably infinite number of different infinite horizon action sequences. However, set of all possible action sequences that can be generated by a robot with a Turing Machine as brain is countably infinite. This is due to the fact that the set of all Turing computable reals is enumerable and countably infinite \cite{turing1936a}. Hence, a humanoid robot with a Turing Machine brain cannot perform all possible humanly doable tasks in any humanly operable environment. These results point that human level general intelligence is not achievable in any Turing Machine. The ability of human beings to perform sophisticated task sequences in arbitrary environments - a fundamental aspect of human intelligence cannot be emulated by any robotic system with a Turing Machine brain.

It is important to note the connection between this thought experiment and evolution. The set of all possible environments - the set of all possible forests, valleys and caves, is uncountably infinite. Hence, species equipped with a classical computing system cannot operate in in all possible environments and perform tasks that are required for their survival. Therefore, the control unit of humans/other species should be able to generate sequences that are uncomputable for Turing Machines.
\section{Quantum Wave function collapse and \textit{Abhi} Machines:}
The fact that Turing Machines cannot emulate human intelligence leads us to the following question:
What gives the human brain hypercomputing abilities?
In particular, how is that the human brain is able to compute sequences that are uncomputable to Turing Machines? 

One possibility is that the human brain is a real computer, that is, the neurons are capable of performing computations with infinite precision real numbers.   In \cite{siegelmann1994analog}, Siegelman et.al. have shown that neural networks, with real valued weights can compute non Turing computable functions. However, because of the Bekenstein bound \cite{bekenstein1981universal}, it is not possible to store real numbered values in a physical substrate and use it directly for computations. Further arguments against real valued computations are discussed in \cite{aaronson2005guest}.

Another proposal, suggested by Penrose in \cite{penrose1999emperor} is that the brain is a quantum computer, and quantum coherence in microtubules \cite{hameroff1996conscious} enable quantum computing in brain. However, any quantum computer can be simulated by a Turing Machine, hence quantum computation by itself does not guarantee hypercomputing abilities. 

I propose a novel form of quantum computing machine, called the \textit{Abhi} Machine. \textit{Abhi} machines are similar to orchestrated objective reduction model \cite{hameroff1996conscious}, in the sense that computation in an \textit{Abhi} machine involves sequential collapse of a quantum wave function.  

An \textit{Abhi} Machine contains a two state quantum system with a wave function $\Psi(t)$ evolving as per the schrodinger equation. The input tape $I_N$ to the \textit{Abhi} Machine contains a sequence of $N$ rational numbers.
\begin{equation}
  I_N =  |*|t_1|t_2|t_3|t_4|t_5|t_6|t_7|t_k|....|t_N|*|
\end{equation}
\begin{equation}
 t_k \subseteq \mathbb{Q}
\end{equation}
At time $t=t_0$, the \textit{Abhi} Machine starts with an initial state $\ket{\Psi_0}$, and evolves as per the schrodinger equation.
\begin{equation}
\ket{\Psi(t)} = U(t)\ket{\Psi_0}
\end{equation}
            where 
\begin{equation}
U(t) =  e^{-iHt}
\end{equation}
$H$ is the hamiltonian of the system and $H_{12} \neq 0$.\\
At time $t=t_1$, the system is measured and the wave function collapses to the state $\ket{\Psi_1}$. $\ket{\Psi_1}$ is recorded as the first digit of the output sequence $O_N$. The quantum system is again evolved as per the schrodinger equation, and the system is measured at time $t = t_2$. The observed state at time $t = t_2$ is recorded as the second digit of output sequence $O_n$. The process continues until the entire input tape is read, or when time $t = t_N$. The output tape at the end of the process is:
\begin{equation}
  O_N =  |*|\ket{\Psi_1}|\ket{\Psi_2}|\ket{\Psi_3}|\ket{\Psi_4}|\ket{\Psi_5}|\ket{\Psi_6}|\ket{\Psi_7}|\ket{\Psi_k}|....|\ket{\Psi_N}|*|
\end{equation}
As this is a two state quantum system, the measured states can be represented by $\ket{0}$ or $\ket{1}$.
\begin{equation}
  \ket{\Psi_k} \in \{\ket{0},\ket{1}\}
\end{equation}
Each quantum measurement, during the operation of an \textit{Abhi} Machine is fundamentally random and non algorithmic. Corresponding to a particular input tape $I_N$, the output $O_N$ of an \textit{Abhi} Machine is unpredictable by any algorithm. The input sequence $I_N$ controls the probability of a particular output sequence $O_N$. For an input tape of size $N$, the output tape $O_N$ could be one among $2^N$ different sequences. Because of the inherent random nature associated with the quantum measurement process, \textit{Abhi} Machines can compute decimal expansion of \textit{any} real number. The set of all Turing computable numbers are however limited, and countably infinite. Because of this, an \textit{Abhi} Machine is computationally more powerful than a Turing Machine and does not obey the Church-Turing thesis.

\begin{theorem}
Abhi Machines can compute sequences that are uncomputable for a Turing Machine.
\end{theorem}
\begin{proof}
Let $I_N$ be an input tape of size $N$ to an Abhi Machine $AM$. For a particular input tape $I_N$, an Abhi Machine $AM$ can generate upto $2^N$ output sequences. Now consider the set of possible input sequences $I = \{I_1,I_2,,,,I_k\}$\\ containing k different input sequences. Let the set of all possible output sequences corresponding to $I$ be $O = \{O_1,O_2,,,,O_j\}$. The cardinality of set $O$ is $2^k$.
\begin{equation}
   |O| = 2^k
\end{equation}
The set of all possible input sequences $I_{\aleph_{0}}$ is countably infinite and hence as:
\begin{equation}
  k\rightarrow \aleph_{0}
\end{equation}
The cardinality of the output set $|O|$ is equal to:
\begin{equation}
  |O| \rightarrow 2^{\aleph_{0}}
\end{equation}
This means that an \textit{Abhi} Machine can compute decimal expansion of any real number. However, the set of all Turing computable reals is countably infinite. Hence, \textit{Abhi} Machines can compute sequences that are uncomputable for a Turing Machine.
\end{proof}
\textit{Abhi} Machines can be used in reinforcement learning agents \cite{sutton2018reinforcement} for sampling action outputs from a probabilistic policy $\pi$. An agent with an \textit{Abhi} Machine would be able to compute infinite horizon action sequences that are uncomputable for an agent that uses a Turing Machine for action selection. Because of this, an agent that uses \textit{Abhi} Machine for action selection would be able to reach states that are unreachable for an agent that uses Turing Machine for action selection. Let's prove this statement mathematically by considering a one dimensional environment with deterministic linear dynamics.\\ 

\begin{theorem}

Given two agents $A_{TM}$ and $A_{AM}$, where $A_{TM}$ denotes an agent that uses Turing Machine for computing actions and $A_{AM}$ denotes an agent that uses \textit{Abhi} Machine for computing actions, agent $A_{AM}$ can reach states that are unreachable for $A_{TM}$.
\end{theorem}
\begin{proof}
Consider a one dimensional environment with determinstic linear dynamics. 

\begin{equation}
    \dot{x} = \alpha u
\end{equation}
where $k$ is a constant, $u$ denotes the action selected by the agent, and $x$ denotes the state of the agent in the environment. Further, $u \in \{0,1\}$. Now consider two agents $A_{TM}$ and $A_{AM}$ living in the same environment that selects actions at each timestep $t_k$. We consider a special problem setup where the time difference between action selections varies over the entire time horizon. For instance, the agent applies $u_0$ at time $t_0 =0$, until time $t = t_1 = \frac{1}{2}$ then applies $u_1$ at time $t_1 = \frac{1}{2}$, until time $t = t_2$ and then applies $u_2$ at time $t_2 = \frac{1}{4}$ and so on. In this manner, the total time horizon $T = 1$, but an agent performs an infinite sequence of actions $\{u_k\}_{k=0}^{\infty}$ in the time horizon $T$. Let $\delta t_k$ denote the time difference between timesteps $t_{k+1}$ and $t_k$.
\begin{equation}
   \delta t_k = t_{k+1} - t_k 
\end{equation}
 $\delta t_{k+1} = \frac{\delta t_k}{2}$ and  $\delta t_{0} = \frac{1}{2}$.
\begin{equation}
    \forall t \in [t_k, t_{k+1}), u(t) = u_k
\end{equation}
and 
\begin{equation}
   u_k \in \{0,1\}
\end{equation}
When $u_k =0$, the dynamics equation becomes:
\begin{equation}
 \dot{x} = 0,  \forall t \in [t_k, t_{k+1})
\end{equation}
Upon integrating the dynamics equation, the state $x$ at timestep $t_{k+1}$, $x_{k+1} = x_k$.

When $u_k =1$, the dynamics equation becomes:
\begin{equation}
 \dot{x} = \alpha,  \forall t \in [t_k, t_{k+1})
\end{equation}
Upon integrating the dynamics equation, the state $x$ at timestep $t_{k+1}$, $x_{k+1} = x_k + \alpha (t_{k+1} - t_k)$.\\ 

The dynamics equation for both cases ($u_k = 0$ and $u_k = 1$) can be combined and expressed as:
\begin{equation}
   x_{k+1} = x_k + \alpha .\delta t_k. u_k 
\end{equation}
After applying an infinite sequence of actions $\{u_k\}_{k=0}^{\infty}$, the agent would reach state $x_{\infty} = x(T)$ at time $t = T = 1$.
\begin{equation}
  x(T) = x_0 + \alpha \Sigma_{k=0}^{\infty} \delta t_k u_k
\end{equation}
As $\delta t_{k+1} = \frac{\delta t_k}{2}$ and  $\delta t_{0} = \frac{1}{2}$, eq(41) becomes:

\begin{equation}
  x(T) = x_0 + \alpha \Sigma_{k=0}^{\infty} \frac{1}{2^{k+1}} u_k
\end{equation}
Now we consider two agents $A_{TM}$ and $A_{AM}$, where $A_{TM}$ performs an infinite horizon action sequence that is precomputed by a Turing Machine and  $A_{AM}$ performs an infinite horizon action sequence that is precomputed by an \textit{Abhi} Machine. Let $x_0 = 0$,$\alpha = 1$, so the state after time $T$, $x(T)$ is given by:
\begin{equation}
  x(T) =  \Sigma_{k=0}^{\infty} \frac{1}{2^{k+1}} u_k
\end{equation}
As $u_k \in \{0,1\}$, $x(T)$ converges to a value between $0$ and $1$. This is due to the fact that the R.H.S term is the expression for binary expansion of a real number between $0$ and $1$. The value of $x(T)$ depends on the action sequence $\{u_k\}_{k=0}^{\infty}$. Let $  x(T)_{A_{TM}}$ denote the terminal state for an agent that uses Turing Machine for computation of actions and let $  x(T)_{A_{AM}}$ denote the terminal state for an agent that uses \textit{Abhi} Machine for computation of actions. If the action sequence is computed from a Turing Machine, then the $x(T)$ would be a computable real number between $0$ and $1$. Since the set of all computable real numbers is enumerable, this means that $x(T)$ cannot acquire an uncountably infinite number of state values. In other words, only a subset of states between 0 and 1 is reachable for an agent that uses Turing Machine for computing actions.

\begin{equation}
  x(T)_{A_{TM}} \in X_{TM}\subseteq [0,1]
\end{equation}
where $X_{TM}$ denotes the enumerable set of states that are reachable for an agent that uses Turing Machine for computing actions.
\begin{equation}
  x(T)_{A_{AM}} \in X_{AM} = [0,1]
\end{equation}
where $X_{AM}$ denotes the enumerable set of states that are reachable for an agent that uses \textit{Abhi} Machine for computing actions. Clearly, $X_{TM} \subseteq X_{AM}$ and hence an agent that uses \textit{Abhi} Machine for computing actions can reach states that are unreachable for an agent that uses Turing Machine for computing actions.\\

The above result can also be extended to multidimensional nonlinear dynamical systems.
\end{proof}
The \textit{Theorem 6.2} indicates that \textit{Abhi} Machines and Turing Machines are not mathematically equivalent, as it is possible to distinguish between the two in thought experiments. Since \textit{Abhi} machines can compute decimal expansion of any real number, agents equipped with \textit{Abhi} machines as control unit can perform tasks that cannot be performed by agents equipped with Turing Machines as control unit. From the thought experiment in Section 5, we saw that to perform tasks in all possible environments, the agent should be to be able to generate an uncountably infinite number of different infinite sequences. This is possible only if the action selection of the agent is due to some form of quantum wave function collapse, similar to the functioning of \textit{Abhi} Machines. Hence, I argue that the ability of the human mind to compute action sequences for performing tasks in an uncountably infinite number of environments is due to quantum computation in the human brain.\\

\section{ A Note about Consciousness}
From the above results, it is clear that the randomness induced by a quantum measurement process could enable hypercomputation. The chemical reactions in the human brain, at a neuronal level could be sources of quantum noise, and this could be a plausible explanation for the hypercomputing abilities of human mind. The human brain is not a simple \textit{Abhi} Machine that generates a random sequence as output. A plausible scenario is that each action selection might involve collapse of a quantum wave function, and the probability of selection of the action is controlled by some classical computation. It is unclear if consciousness is connected to the wave function collapse process as Penrose has suggested in \cite{penrose1994shadows}. This is quite plausible, however since there is no mathematical definition of consciousness, it is difficult to assert the claim that consciousness emerges due to quantum wave function collapse. 

\section{Discussion}

These results invalidate the Church-Turing thesis, and they also imply that no digital computing system  can reach human level intelligence. For achieving human level general intelligence, the computing system should have some form of quantum computation, involving sequential quantum wave function collapse. These findings, point us to the right direction for building human level general intelligent systems.

\section{Contributions}
\begin{itemize}
 \item Provided Strong arguments against Church Turing thesis.
  \item Showed that quantum computation is \textbf{required} for developing systems with human level general intelligence.
\end{itemize}

\section*{Acknowledgments}
I would like to thank Prof. Sir Roger Penrose for email conversations. 

\bibliographystyle{unsrt}  
\bibliography{references}

\end{document}